\documentclass[11pt]{article}
\usepackage[utf8]{inputenc} 
\usepackage{parskip}
\usepackage{amssymb, amsmath, amsthm, graphicx, subfigure}
\usepackage{enumerate}
\usepackage[dvipsnames]{xcolor}
\usepackage[colorlinks=true,linktoc=page]{hyperref}
\usepackage{braket}
\usepackage{esvect}
\usepackage[margin=1in]{geometry}
\usepackage{mathtools}
\usepackage{tikz}
\usepackage{algorithm}
\usepackage{algorithmicx}
\usepackage[noend]{algpseudocode}

\usepackage{mdframed}

\usetikzlibrary{patterns,patterns.meta}

\theoremstyle{plain}
\newtheorem{theorem}{Theorem}[section]
\newtheorem{lemma}[theorem]{Lemma}

\theoremstyle{definition}

\theoremstyle{remark}

\newcommand{\R}{\mathbb{R}}
\newcommand{\one}{\mathbf{1}}

\newcommand{\tr}{\operatorname{tr}}

\newcommand{\diag}{\operatorname{diag}}

\newcommand{\Aobj}{\mathsf{A}}
\newcommand{\Dobj}{\mathsf{D}}
\newcommand{\Eobj}{\mathsf{E}}

\title{Hardness of A/E-Design under Partition Constraints}

\author{Nikhil Bansal\thanks{University of Michigan. Email: \texttt{bansaln@umich.edu}. Supported by NSF awards CCF-2327011 and CCF-2504995.} \and Yuze Xu \thanks{University of Michigan. Email: \texttt{yuzex@umich.edu}}}
\date{}
\begin{document}
\maketitle
\vspace{-2em}

\begin{abstract}
We consider the A/E-design problem under partition constraints: Given vectors $v_1,\ldots,v_N\in \R^d$ and a partition matroid on $[N]$, find a base $S$ of the matroid that minimizes 
$\tr(M(S)^{-1})$ or $\lambda_{\max}(M(S)^{-1})$ where $M(S)=\sum_{i \in S} v_i v_i^\top$.

In contrast to D-design, where good estimation and approximation guarantees are known as a function of $d$ \cite{MadanNikolovSinghTantipongpipat2020,BLPST22, BLPS22}, we show that no reasonable approximation exists for A/E-design. This answers a question of Brown, Laddha and Singh \cite{BLS24}. The proof is based on an elementary reduction from three-dimensional matching.
\end{abstract}

\section{Introduction}

In an experimental design problem, we are given vectors
$v_1,\ldots,v_N\in\R^d$. The goal is to choose a subset $S \subset [N]$, subject to some combinatorial constraints, so that the resulting matrix  
$M(S):=\sum_{i\in S}v_iv_i^\top$ is far from singular. 
 
This is a central problem in statistics \cite{Pukelsheim2006}, where
the vectors $v_i$ correspond to measurements of an unknown linear model, and
$M(S)^{-1}$ corresponds to the covariance of the
least-squares estimator. More recently, it has also found several applications in
machine learning, numerical linear algebra and graph algorithms, see e.g., \cite{SinghXie2018, AZLSW21, NST22, LZ22, Lau0025} and references therein. 

The constraints on the subset
$S$ are typically cardinality-based ($|S|= k$) or involve more general matroid constraints.
The three widely studied optimization criteria for $M(S)$ are D/A/E-design.  We use their
minimization forms
\[
\Dobj(S):=\det(M(S)^{-1})^{1/d},\qquad
\Aobj(S):=\tr(M(S)^{-1}),\qquad
\Eobj(S):=\lambda_{\max}(M(S)^{-1}),
\]
with infinite cost when $M(S)$ is singular.  If
$\mu_1,\ldots,\mu_d$ are the eigenvalues of $M(S)^{-1}$, these objectives are,
up to normalization, the power means of the $\mu_i$ for
$p=0,1,\infty$, respectively.

\paragraph{Cardinality constraints.}
The cardinality setting, where $S$ is any arbitrary subset of size $k\ge d$, has been studied
extensively, with non-trivial results for all $D/A/E$-design objectives. When $k=d$, the best known approximation factors are $e$ for
D-design, $d$ for A-design, and $O(d^2)$ for E-design \cite{NST22}.  For larger $k$, improved guarantees have been obtained using a variety of techniques such as proportional volume sampling  \cite{SinghX20,NST22},  regret-minimization \cite{AZLSW21}
and local-search \cite{MadanSinghTantipongpipatXie2019,LZ22}. Recently, Lau, Wang, and Zhou \cite{Lau0025} give an elegant unified way to obtain the known bounds using the interlacing polynomial method of Marcus, Spielman and Srivastava \cite{MSS15a, MSS15b, MSS22}.

\paragraph{Matroid constraints and D-design.}
For more general matroid constraints, a rich theory has been developed  for D-design, based on ingenious relaxations and rounding techniques. In a seminal work, Nikolov and Singh \cite{NS16} considered partition matroids, and introduced a novel max--min geometric program
and gave an $e$-estimation algorithm for $k=d$ using stable-polynomial methods. These results were generalized in an interesting line of
work based on real-stable and log-concave polynomials
\cite{AG17,AGV18} to give an $O(1)$-estimation algorithms for arbitrary matroids, for $k=d$. Later, for the harder $k>d$ setting, 
\cite{MadanNikolovSinghTantipongpipat2020} obtained an $O(1)$-estimation algorithm for partition matroids, and a $\mathrm{poly}(d)$-estimation for arbitrary matroids. Algorithms that return an actual solution with $\mathrm{poly}(d)$ guarantee have also been obtained more recently \cite{BLPST22,BLPS22}.

Given this remarkable progress for D-design from cardinality constraints to more general matroid constraints, it is natural to ask if analogous progress is also possible for A- and E-design. This question was explored recently by Brown, Laddha and Singh \cite{BLS24}, who studied E-design with partition constraints. They showed that the problem  has interesting connections to the Kadison--Singer \cite{MSS15b} and Santa Claus problems \cite{BansalS06, ChakrabartyCK09},
and gave non-trivial results in the special case when $d$ is fixed, or when all the vectors have small leverage scores. 
Motivated by these connections and results, they ask whether the problem admits reasonable approximation algorithms in general.

\subsection{Our result}
In this work, we answer the question above  negatively and show the following hardness result.

\begin{theorem}
\label{cor:precision}
Let $d$ denote the dimension, and $B$ the maximum bit length of any coordinate\footnote{All coordinates with be integers in our construction.}
of an input vector.
Assuming $\mathsf{P} \neq \mathsf{NP}$, 
for any
$\varepsilon\in(0,1)$, there is no efficient algorithm with approximation guarantee $2^{\mathrm{poly}(d)+(1-\varepsilon)B}$  for A/E-design with partition constraints.
\end{theorem}
Thus in contrast to D-design, the A/E-design problem becomes substantially harder beyond cardinality constraints, and is unlikely to admit any reasonable approximation algorithms.

The proof of Theorem \ref{cor:precision} is elementary, and is based on a reduction from the classical three-dimensional matching problem (3-DM), which is well-known to be NP-complete \cite{Karp72}. 

\section{Preliminaries}
We begin by describing the problem formally.

We are given vectors
$v_1,\ldots,v_N\in\R^d$, and a partition matroid on $[N]$ with parts
$P_1,\ldots,P_\ell$, and positive integers $b_1,\ldots,b_\ell$.
A base in this matroid is any set $S$ satisfying $
|S\cap P_a|=b_a$ for every $a\in[\ell]$. We call $k=\sum_a b_a$ the rank of the matroid, and assume that $k\geq d$, as $M(S)$ is always singular if $k<d$. 

Our hard instances will consist of $d$
parts with $b_a=1$ for every part, and thus $k=d$. Notice that the same hardness holds when 
$k>d$, as one may add extra parts containing zero vectors, without affecting anything.

For any subset $S$ of $d$ vectors,  let
$V_S \in\R^{d\times d}$
be the matrix with columns $v_i$ for $i\in S$, and let
\[M(S):= \sum_{i\in S} v_i v_i^\top =  V_SV_S^\top.\]
The A/E objectives can be expressed in terms of $ V_S$ as
\begin{equation}\label{eq:objectives-inverse}
\Aobj(S)= \tr (M(S)^{-1}) = \|V_S^{-1}\|_F^2,
\qquad
\Eobj(S)=  \|M(S)^{-1}\|_{\mathrm{op}} = \|V_S^{-1}\|_{\mathrm{op}}^2.
\end{equation}
Here and throughout, we assume that $V_S$ and $M(S)$ are always invertible, as these objectives are trivially infinite otherwise. Note that we always have 
$\Eobj(S)\le \Aobj(S)\le d\,\Eobj(S)$, and thus it suffices to show hardness for either of these objectives.

\paragraph{Geometric Interpretation.} 
Both $\Eobj(S)$ and $\Aobj(S)$ in \eqref{eq:objectives-inverse} also have a nice geometric interpretation, which will be useful to understand our idea more intuitively. 

Fix some set $S \subset [N]$ of size $d$. Suppose we want to
express a target vector $h \in \R^d$ with the columns in $V_S$ as the basis, so that $h = V_S\, y$. Then the unique coefficient
vector is
$y:= y(h) = V_S^{-1}h$, and thus
\begin{equation}
\label{eq:geom}
\Eobj(S)=\max_{\|h\|_2=1}\|V_S^{-1}h\|_2^2 = \max_{\|h\|_2=1} \|y(h)\|_2^2,
\end{equation}
measures the worst possible blow up of coefficients, in a $\ell_2$ sense, over all unit vectors. Similarly, one has that $\Aobj(S)=\sum_{j=1}^d\|V_S^{-1}u_j\|_2^2$, for every orthonormal basis $u_1,\ldots,u_d$.

This view is useful to certify a large lower bound for $\Eobj(S)$, by simply exhibiting some
unit vector $h$ for which the vector  $V_S^{-1}h$ is large.  

\paragraph{3-dimensional Matching (3-DM).}
In an instance $\mathcal{I}$ of 3-DM, we are given three disjoint
sets $X,Y,Z$, each of size $q$, and a family of triples
$\mathcal{T}\subseteq X\times Y\times Z$.  The goal is to determine whether there are
$q$ triples in $\mathcal{T}$ that cover each element of $X\cup Y\cup Z$ exactly once. We call this a perfect matching.

Let us denote $m=|X\cup Y \cup Z|=3q$, and $n=|\mathcal{T}|$. Let
$A\in\{0,1\}^{m\times n}$ denote the incidence matrix of the triples. Then
each column of $A$ has exactly three ones and a perfect matching exists iff
there is an $x\in\{0,1\}^n$ such that
\begin{equation}\label{eq:exact-cover}
Ax=\one_m.
\end{equation}

\section{The Reduction and Hardness}

Given an instance $\mathcal{I}$ of 3-DM, we will construct a collection of $N=m+2n+1$ vectors in $\R^d$ with $d=m+n+1$, and a partition matroid on $[N]$ with rank $d$. 

As usual, we want the following two properties from the reduction:  

\,\,\,\,\, (i) Completeness: For every perfect matching in $\mathcal{I}$, there is some base $S$ with $\|V_S^{-1}\|_{\mathrm{op}}$ small. 

\,\,\,\,\, (ii) Soundness: If $\mathcal{I}$ has no perfect matching, $\|V_S^{-1}\|_{\mathrm{op}}$ is large for every base $S$.

\medskip

The high-level idea is the following.

 By \eqref{eq:exact-cover}, $\mathcal{I}$ has a perfect matching iff the vector $Ax-\one_m = \mathbf{0}$ for some $x \in \{0,1\}^n$. 
 In our reduction, the possible bases $S$ in the matroid will be in bijection with $x \in \{0,1\}^n$, and the vectors $v_i$ will be designed so that the matrix $V_S^{-1}$ contains the vector $R^2(Ax-\one_m)$ inside it, where $R$ is some scalar, while the other entries of $V_S^{-1}$ have magnitude $O(R)$.
 By making $R$ arbitrarily large, 
 we obtain the desired hardness gap.
 
We now give the details.

\subsection{The Construction}
Let
 the dimension $
d:=n+m+1$.
We divide the standard coordinate basis for $\R^d$, into three pieces:
\[e_1,\ldots,e_m,\ u_1,\ldots,u_n,\ h.
\]
Here
 each $e_j$ for $j\in [m]$ will correspond to an element in $X \cup Y \cup Z$, each $u_i$ for $i\in [n]$  to a triple in $\mathcal{T}$, and finally, the coordinate $h$ will be special.

There are three types of vectors and parts:
\begin{enumerate}
    \item (Element Part.) For every $j\in[m]$, there is a part containing only a single basis vector $e_j$. So there is no choice here.
\item (Triple Part.)
Fix an integer $R\ge1$.  For each triple $i\in[n]$, create a part 
$P_i:=\{u_i,\ u_i+Ra_i\}$,
where $u_i$ is the basis vector above, and $a_i:=\sum_{j=1}^m A_{ji}e_j$ is
the indicator of the triple $i$.
\item (Special Part.)
A part with the single vector \[w:=h+R\sum_{i=1}^n u_i+R^2\sum_{j=1}^m e_j.\]
\end{enumerate}
Notice that there are exactly
$d=m+n+1$ parts, and $m+2n+1$ vectors in $\R^d$.

\subsection{Analysis} 
Any base $S$ is completely determined by specifying whether the first or the second vector is chosen in each part $P_i$, for $i\in [n]$, corresponding to triples.
So each base $S$ is in bijection with $x\in\{0,1\}^n$, where we choose the second
vector from $P_i$ if $x_i=1$, and the first if $x_i=0$. 
Using this notation, the chosen vector from $P_i$ can be written as   
\begin{equation}
    \label{eq:gi}
    g_i:=u_i+Rx_i a_i = u_i +R x_i \sum_{j=1}^m A_{ji}e_j.
\end{equation}
 The key point of this construction is that to express $h$ using the vectors in $S$ encodes
 $Ax-\one_m$. In particular, for any base $S$ (equivalently $x \in\{0,1\}^n$), we have a unique representation of $h$ using vectors in $S$:
\begin{equation}\label{eq:witness-identity}
h 
=w-R\sum_{i=1}^n g_i
 +R^2\sum_{j=1}^m\bigl((Ax)_j-1\bigr)e_j.
 \end{equation}
Indeed, the coefficients in this representation are forced: the $h$ coordinate only appears in  $w$ and hence the coefficient $w$ must be one. Now, to cancel the terms $Ru_i$ in $w$, as each $u_i$ only appears in $P_i$, the coefficient of each $g_i$ is forced to $-R$. Finally, the residual term follows as
\begin{align} w - h - R \sum_{i=1}^n g_i =  R^2 (\sum_{j=1}^m  e_j - \sum_{i=1}^n \sum_{j=1}^m A_{ji} x_i e_j) = R^2 \sum_{j=1}^m  \bigl(1-(Ax)_j\bigr)e_j  \end{align}
where the first equality uses \eqref{eq:gi}.

So if $Ax=\one_m$, then the residual disappears, and $h$ is generated using
coefficients of magnitude at most $R$.  If $Ax\ne\one_m$, then
$Ax-\one_m$ is a nonzero integer vector, so the representation of $h$
has magnitude at least $R^2$.
By the geometric interpretation in \eqref{eq:geom}, this already gives an $\Omega(R^4)$ E-cost in the NO case. For the YES case,
we must also control the representations of all other directions; this follows
from the exact inverse below.

\subsection*{The inverse matrix and the gap}
Fix a base $S$, specified by some $x \in \{0,1\}^n$, and 
let $A_x:=A\diag(x) \in \R^{m\times n}$.  Ordering the selected columns as
$g_1,\ldots,g_n,w,e_1,\ldots,e_m$,
we can write the matrix $V_S  = V_x $ in the block form
\begin{equation}\label{eq:Vx}
V_x=
\begin{pmatrix}
I_n & R\one_n & 0\\
0   & 1        & 0\\
RA_x & R^2\one_m & I_m
\end{pmatrix}.
\end{equation}
Here the top left, bottom right and middle blocks  have dimensions $n\times n$, $m\times m$ and $1\times 1$ respectively.
\begin{lemma}[Exact inverse]\label{lem:inverse}
For every $x\in\{0,1\}^n$,
\begin{equation}\label{eq:Vx-inverse}
V_x^{-1}=
\begin{pmatrix}
I_n & -R\one_n & 0\\
0   & 1         & 0\\
-RA_x & R^2(Ax-\one_m) & I_m
\end{pmatrix}.
\end{equation}
\end{lemma}
\begin{proof}
This can be verified directly by
multiplying \eqref{eq:Vx} and \eqref{eq:Vx-inverse} and noting that $A_x\one_n=Ax$. Every block is immediate
except the lower middle block, which is
\[
-R^2A_x\one_n+R^2\one_m+R^2(Ax-\one_m)=0. \qedhere
\]
\end{proof}

Taking the squared Frobenius norm in \eqref{eq:Vx-inverse} gives the following
exact expression.

\begin{lemma}[Completeness and soundness]\label{lem:gap}
For every $x\in\{0,1\}^n$,
\begin{equation}\label{eq:A-identity}
\Aobj(x) = \|V_x^{-1}\|_F^2
=d+nR^2+R^2\|A_x\|_F^2+R^4\|Ax-\one_m\|_2^2.
\end{equation}
In particular if $Ax=\one_m$, then
$\Eobj(x)\le \Aobj(x)\le 2dR^2$. Else if
 $Ax\ne\one_m$, then
$\Aobj(x)\ge \Eobj(x)\ge R^4$.
\end{lemma}

\begin{proof}
Equation \eqref{eq:A-identity} is obtained by summing the squared entries of
\eqref{eq:Vx-inverse}.  
As each column of $A$ has three ones,
$\|A_x\|_F^2=3\sum_i x_i$.
If $Ax=\one_m$, summing all coordinates gives $3\sum_i x_i=m$.
Substituting this into \eqref{eq:A-identity} gives the first claim.

If $Ax\ne\one_m$, the vector $Ax-\one_m$ is a nonzero integer
vector, and has squared norm at least one.  The column of
$V_x^{-1}$ corresponding to coordinate $h$ has lower block
$R^2(Ax-\one_m)$, which has squared Euclidean norm is at least $R^4$, and hence
\[
\Eobj(x)=\|V_x^{-1}\|_{\mathrm{op}}^2\ge R^4.
\qedhere \]
\end{proof}

\begin{proof}[Proof of Theorem~\ref{cor:precision}]
Fix any $\varepsilon \in (0,1)$ and integer $c>0$, and apply the construction above with $R:= 2^{d^{c+1}}$.
The largest coordinate in the construction is $R^2$, and thus $B=2d^{c+1} +O(1)$.

By Lemma \ref{lem:gap}, the ratio between optimum values for the NO and YES instances is at least $R^2/(2d)
=2^{B-O(\log d)}$.
For $d$ large enough, we
have
$d^c+(1-\varepsilon)B < B-O(\log d)$.
As the coordinates have $B=O(d^{c+1})$ bits, the
encoding length and the running time of the reduction are polynomial in the size of the $3$-DM instance. As any polynomial is eventually dominated by $d^c$ for suitable $c$, we obtain that, assuming $\mathsf{P}\neq \mathsf{NP}$,
no polynomial time $2^{\mathrm{poly}(d)+(1-\varepsilon)B}$-approximation exists for A/E-design.
\end{proof}

\section{Declaration of Generative AI Use}
The authors used GPT-5.6 Pro during the development of this
work to simplify their original construction and assist with verification. GPT was not used in any part of the exposition, and the authors take full responsibility for the content and correctness.

\bibliographystyle{alpha}   
\bibliography{ref}
\end{document}